\documentclass[12pt,titlepage]{article}

\def\S{\mathcal{S}}

\def\Re{\mathbf{R}}

\def\ep{\varepsilon}
\def\w{\omega}

\def\da{\delta}
 
\def\phi{\varphi}

\def\ul{\underline}

\def\os{\emptyset}

\newcommand{\df}[1]{\textit{#1}}

\newcommand{\abs}[1]{ \left | #1 \right | }

\usepackage{xcolor}
\usepackage{tikz}
\usetikzlibrary{shapes, fit}

\usepackage{mathtools, latexsym, amsmath, amssymb,amsthm, natbib,marvosym, verbatim,setspace}
\usepackage{amssymb}
\usepackage[utf8]{inputenc}
\usepackage[a-1b]{pdfx}
\usepackage[margin=1.5in]{geometry}

\newtheorem{theorem}{Theorem}
\newtheorem{proposition}{Proposition}
\newtheorem{lemma}{Lemma}

\theoremstyle{remark}
\newtheorem*{remark}{Remark}
\newtheorem{example}{Example}

\setcitestyle{authoryear}
\def\cite{\citet}

\title{Stable Allocations in \\ Discrete Exchange Economies\thanks{We thank Fuhito Kojima, John Nachbar, and Leo Nonaka for helpful comments.}}
\author{Federico Echenique\thanks{Department of Economics, UC Berkeley, {fede@econ.berkeley.edu}. Echenique thanks the National Science Foundation for its support through the grants SES 1558757 and  CNS 1518941.}, Sumit Goel\thanks{Division of Social Sciences, New York University, Abu Dhabi, {sumitgoel58@gmail.com}}, and SangMok Lee\thanks{Department of Economics, Washington University in St. Louis, {sangmoklee@wustl.edu}}}
\date{\today}

\begin{document}

\maketitle

\begin{abstract}
We study stable allocations in an exchange economy with indivisible goods. The problem is well-known to be challenging, and rich enough to encode fundamentally unstable economies, such as the roommate problem. Our approach stems from generalizing the original study of an exchange economy with unit demand and unit endowments, the \emph{housing model}. Our first approach uses Scarf's theorem, and proposes sufficient conditions under which a ``convexify then round'' technique ensures that the core is nonempty. The upshot is that a core allocation exists in categorical economies with dichotomous preferences. Our second approach uses a generalization of the TTC: it works under general conditions, and finds a solution that is a version of the stable set. 
\end{abstract}

\section{Introduction}

Economists have a good understanding of stable allocations in discrete exchange economies under some popular, but very special, assumptions. The literature with quasi-linear preferences, or transfers, is highly developed. Almost everything we know about discrete exchange economies is in the literature on auctions and pricing, under the assumption of quasi-linear preferences. The literature refers to these models as markets ``with money,'' and money is ubiquitous in mechanism design.\footnote{The concept of money is often misunderstood. Of course, actual economies feature money, but in the models with transfers, money is treated as a consumption good, and one that enters utility linearly. It is not paper money, which typically needs some sort of trading friction to have value. Money as consumption good should be thought of as a summary of consumption in the rest of the economy; and even under this interpretation, the quasi-linear assumption rules out risk aversion and income effects.} 
Without transfers, most progress is limited to models of unit demand; the so-called housing market of \cite{shapley1974cores}.\footnote{We focus on one-sided, object allocation, models. Not on two-sided markets.} The general discrete multi-good allocation problem is, however, known to be very difficult. The present paper is an attempt at furthering our understanding of this difficult problem. 

The discrete exchange economy is important to understand for conceptual and
theoretical reasons, and because it covers important practical applications. First,
theory. The exchange economy is our most basic model of trade, in which agents' motives for mutually advantageous trade are captured. The model is very well
understood (and taught to every student of economics) under the technical assumption of infinitely divisible goods. Indeed, in an economy with infinitely divisible goods, standard assumptions of convexity and continuity suffice to establish the existence of various solution concepts (see chapters 15-17 in \cite{mwg1995}). Many important questions about the structure of equilibria, connection between different solution concepts and their welfare implications, as well as the scope of general equilibrium theory, are all well understood. Without the assumption of infinitely divisible goods, much less is known about the basic model of exchange. So we think that it is conceptually very important to better understand models of discrete multi-goods markets. Put simply (if dramatically): \textit{The profession's understanding of markets and exchange is limited by the extent of our understanding of the general model with indivisible goods.}

Pure theory aside, some important applications rely on a better understanding of the general discrete exchange economy. Perhaps the most glaring class of applications are multi-item auctions for agents that are not risk-neutral over monetary transfers. The most successful practical application is probably to course bidding: \cite*{krishna2008research,sonmez2010course}, \cite{budish2012multi}, \cite*{budish2017course}, and \cite{othman2010finding}.\footnote{An auction is a special case of a discrete exchange economy in which a special agent (seller) is endowed with items and all other agents are endowed with (discrete) money. A course allocation is also an exchange economy if different groups of students are prioritized by different courses by seniority and majors.}

\subsection{Overview of results}

The thrust of our paper is to propose sufficient conditions that ensure the existence of some notion of stability in discrete multi-good exchange economies. The question of the existence and structure of competitive equilibria is obviously important, but not the focus of our paper. We deal exclusively with game-theoretic bargaining paradigms (within discrete exchange economies) and discuss variations of core stability.

Our analysis is based on the two techniques that \cite{shapley1974cores} used to prove
the nonemptiness of the weak core in the housing market. The first technique
establishes that the housing market reduces to a non-transferable utility (NTU) cooperative game satisfying
Scarf's balancedness condition (\cite{scarf1967core}), implying a nonempty weak
core. The second technique, Gale's Top Trading
Cycles (TTC) algorithm,  constructs a core allocation. We generalize the TTC
algorithm to discrete exchange economies.

Using Scarf's theorem, we first establish conditions under which a randomized
allocation is stable: what we call the random core. The random core is of independent
interest, and comprises allocations that would be stable, were it not for the
discrete nature of our model. In the housing model, Shapley and Scarf show that the random core allocation may be ``rounded'' so as to obtain a core allocation. We show that the same rounding technique is not viable in a general exchange economy, but that under some added assumptions, there is a rounding algorithm that delivers a core allocation. 

The upshot is that the core in \textit{categorical} discrete economies is nonempty.  Goods are grouped
into categories, and agents may consume at most one good of each category: The
\textit{house-car-boat model}, to use the language of
\cite{moulin2014cooperative}. Agents have additively-separable and dichotomous
utilities over categorical consumption. Under these assumptions,
Theorem~\ref{thm:categorical_economy} states that the weak core is nonempty. The
proof proceeds by showing that the random core is nonempty, and that our rounding
algorithm constructs a core allocation starting from a random core allocation. 

Our second technique is based on a generalization of the TTC algorithm, using the ideas in
the Tarski-algorithm that has been applied extensively in the literature on stable
two-sided matchings. First we verify that, when applied to a standard housing model,
our algorithm replicates the TTC (actually iteration by iteration) and finds a weak core allocation. Then we consider its behavior in a general discrete exchange economy. We propose a ``locally absorbing set'' with the idea that 
a block to an allocation in the set would trigger a sequence of forward-looking blocks that would end with a (possibly different) allocation in the same set. The set of all allocations is already locally absorbing, so we find a minimal locally-absorbing set.

\subsection{Related Literature}\label{sec:literature}

\cite{shapley1974cores} introduce the model of a housing market, which has been studied very extensively. It is a special case of our model, when agents have unit demands and are endowed with a single good. Their existence proof relies on Scarf's sufficient condition, but they note that a simpler and constructive argument is possible by means of TTC algorithm. The literature following up on Shapley-Scarf, and analyzing the housing market, is huge. \cite{ROTH1977131}, for example, discuss the differences between strong and weak blocks, and clarify the relation between the core of the housing market and the competitive equilibrium allocations. \cite{ma1994strategy} and \cite{sonmez1999strategy} study the incentive properties of core allocations. \cite{sonmez2010house} and \cite*{roth2004kidney} apply the model in practical market-design settings.

The literature on discrete exchange economies is also significant. With no pretense of going through an exhaustive review, we can mention that \cite{henry1970indivisibilites} is mainly focused on the existence problems for competitive equilibrium. He shows that, unless very restrictive assumptions are imposed, an equilibrium is not going to exist. A number of papers seek to overcome these negative results by considering discrete economies in which there exists one perfectly divisible good, a \textit{numeraire}: Perhaps the first paper in this setting is \cite{mas1977indivisible}, who proposes a model in which he can show the existence of competitive equilibria. \cite{quinzii1984core} shows that the core is nonempty in a housing model with unit demand, as well as proving a core equivalence theorem. \cite{svensson1983large}, also imposing unit demand, obtain the existence of allocation with various equilibrium and normative properties. More recently, \cite*{baldwin2020equilibrium} exhibit a connection between a model with transferable utility, and a general model of a market with income effects, which allows them to obtain equilibrium existence results. And \cite{jagadeesan2021matching} proves the existence of quasi-equilibrium allocations, which are then exploited to show that the set of stable allocations is nonempty.

Another line of attack has focused on large economies: keeping the number of goods fixed, while letting the number of agents grow large. See \cite{starr1969quasi} and \cite{dierker1971equilibrium}. \cite{mas1977indivisible} works in the continuum limit (and also, in fact, assumes a perfectly divisible good). More recently, \cite*{budish2011combinatorial} considers discrete markets and shows that existence and incentive properties are alleviated when the number of agents is large. See also \cite*{budish2017course} for a real-world implementation of these ideas to course bidding in business schools.

The model of a categorical economy that we cover in Section~\ref{sec:results} is introduced by \citet{moulin2014cooperative}, who presents as an open question the determination of whether the core is empty. \cite*{konishi2001shapley} provide an answer in the form of a non-existence example (we revisit this example in Section~\ref{sec:results}). Our contribution is to find sufficient conditions on a categorical economy, namely dichotomous preferences for each category of good, under which the core is nonempty.\footnote{Preferences are dichotomous when an agent views each object as either acceptable or unacceptable. Dichotomous preferences are subsumed by lexicographic preferences. Notable applications are compatibility-based organ or blood exchange models \citep{roth2004kidney, han2022blood}.} Closer in spirit to our exercise, 
\cite{sikdar2017mechanism} provides a lexicographically separable preference assumption, for the nonempty core. We assume dichotomous preferences over individual items, while some studies assume dichotomous preferences over \emph{bundles} to obtain a nonempty core \citep{ergin2017dual, nicolo2019matching}. Dichotomous, trichotomous, and even $m$-chotomous preferences over bundles have been considered in the design of mechanisms that satisfy incentive and efficiency properties \citep{sonoda2014two, aziz2020strategyproof, manjunath2021strategy, andersson2021organizing}.

If we think of blocks as objections, our locally absorbing set has a flavor of the bargaining set (\cite{aumann1961bargaining}), which comprises allocations for which any objections are subject to counterobjections. 
There are several different variations on the notion of a bargaining set, such as \cite{mascolellbargaining89} and \cite{zhou1994new}. There are known, general, game-theoretic existence results, such as those of \cite{peleg1967existence} in the setting of transferable utility, and \cite{peleg1963withoutsidepayments} or \cite{VOHRA199119} for games without transfers. Our model does not satisfy the assumption in these papers: we do not have transferable utility, nor are the convexity assumptions in \cite{peleg1963withoutsidepayments}, or the balancedness condition of \cite{VOHRA199119}, satisfied. We also want to emphasize the algorithmic and constructive nature of our existence result for the absorbing set, while the existing literature often uses non-constructive, topological, fixed-point arguments.

\section{Model}\label{sec:model}

An \df{economy} is a tuple $E=(O,\{(v_i,\w_i):i\in A \})$ in which
\begin{itemize}
\item $O$ is a finite set of \df{objects};
\item $A$ is a finite set of \df{agents};
  \item each agent $i\in A$ is described by a utility function
    $v_i:2^O\to\Re\cup\{-\infty \}$ and a nonempty endowment $\w_i\subseteq O$, with
    $O=\cup_i \w_i$ and $\w_i\cap\w_j=\os$ when $i\neq j$.
\end{itemize}

We allow utilities to take on the value $-\infty$ in order to encode
agents' consumption space through the domain of the utility
function. For example, we may consider a \df{shoe economy} in which
agents desire a pair of left and right shoes. The set of objects when we have three
pairs of shoes is:
$O=\{\ell_1,\ell_2,\ell_3,r_1,r_2,r_3 \}$, where $\ell_i$ is a left shoe and
$r_i$ a right one. Now we can say, for example, that
$v_i(\{\ell_1,r_2\})=10>7=v_i(\{\ell_3,r_1\})$, while
$v_i(\{\ell_1,\ell_2\})=-\infty$ says that consuming two left shoes is
not allowed in the model.

More conventionally, the housing model of \cite{shapley1974cores} is a special case of our model. Suppose that $O=\{h_i:i\in A\}$ contains exactly one house $h_i$ for each agent $i$ in $A$; suppose that $\w_i=\{h_i\}$, so that $i$ owns the $i$th house, and  let $v_i(X)=-\infty$ when $X$ is not a singleton subset of $O$.

The following definitions are standard: An \df{allocation} in the economy $E$ is a pairwise-disjoint collection of sets of objects $\{X_i:i\in A\}$ with the property that $\cup_i X_i \subseteq O$. A nonempty subset $S\subseteq A$ is termed a \df{coalition}. For a coalition $S$, an $S$-allocation is a pairwise-disjoint collection of sets of objects $\{X_i:i\in S\}$ with the property that $\cup_{i\in S} X_i \subseteq   \cup_{i\in S}\w_i$. We think of an allocation as the outcome of exchange among the agents in the economy, and of an $S$-allocation as the outcome of exchange among the members of the coalition $S$.

Aside from their use in describing feasible consumption, utility functions summarize agents' ordinal preferences. Given an allocation $X=\{X_i:i\in A\}$, we say that the coalition $S$ \df{weakly blocks} $X$ if there exists an $S$-allocation $\{X'_i:i\in S\}$ with $v_i(X'_i)\geq v_i(X_i)$ for all $i\in S$, and  $v_i(X'_i)> v_i(X_i)$ for at least one $i\in S$. In contrast,  $S$
\df{strongly blocks} $X$ if there exists an $S$-allocation $\{X'_i:i\in S\}$ with $v_i(X'_i)> v_i(X_i)$ for all $i\in S$.

The \df{weak core} of the economy $E$ is the set of allocations that are not strongly blocked by any coalition, while the \df{strong core} is the set of allocations that are not weakly blocked.

Our first example illustrates the difference between the weak and the strong core, and shows how the strong core may be a proper subset of the weak core, even with strict preferences.

\begin{example}[Weak core $\neq$ strong core with strict preferences] 
\label{ex:weak-vs-strong} Consider a shoe
  economy with three agents: $A=\{1, 2, 3\}$ and endowments $\w_i=\{(\ell_i, r_i)\}$,
  $i \in A$. Assume ordinal preferences as follows:
\begin{center}
\begin{tabular}{lll}
Agent 1 & Agent 2     & Agent 3   \\\hline
$(\ell_2, r_3)^X$     & $(\ell_3, r_1)^X$ & $(\ell_1, r_2)^{X,Y}$ \\
$(\ell_3 , r_1)^Y$ & $(\ell_2, r_3)^Y$ & $(\ell_3, r_3)$\\
$(\ell_1, r_1)$     & $(\ell_2, r_2)$ & $\cdots$ \\
$\cdots$ & $\cdots$ & $\cdots  $
\end{tabular}
\end{center}
The table describes the relevant parts of each agent's ordinal preference, or
ranking, over sets of objects.
Agents have no indifferences. 
The table uses superscripts to identify two allocations. The first allocation is 
$X_1=(\ell_2, r_3),$ $X_2=(\ell_3, r_1)$, and $X_3=(\ell_1, r_2)$. The second is 
$Y_1=(\ell_3 , r_1)$,  $Y_2=(\ell_2, r_3)$ and $Y_3=(\ell_1, r_2)$. It is easy to see that the strong core consists only of allocation $X$, while the weak core contains both. Key here is that, in using $X$ to block $Y$, agents $1$ and $2$ need to trade items that belong to agent 3, which requires 3's ``permission.'' A weak block only requires 3's weak preference for them to grant permission, while the strong block insists on a strict preference.
\end{example}

Example \ref{ex:weak-vs-strong} suggests that the weak core can sometimes be too large. To address this issue, some studies have proposed refinements \citep{balbuzanov2019endowments, yilmaz2022stability}. 
However, our paper addresses the opposite problem: the weak core is frequently too small, if not completely empty. Our second example exhibits a shoe economy with an empty weak core.
\begin{example}[Empty weak core in a shoe economy.]\label{ex:shoes}
 $A=\{1, 2, 3\}$ with endowment $\{(\ell_i, r_i)\}_{i=1,2,3}$. Assume preferences such that
\begin{center}
\begin{tabular}{ccc}
Agent 1 & Agent 2     & Agent 3   \\\hline
$(\ell_1, r_2)^X$     & $(\ell_2, r_3)^Y$ & $(\ell_1, r_3)^Z$ \\
$(\ell_3 , r_1)^Z$ & $(\ell_2, r_1)^X$ & $(\ell_3, r_2)^Y$\\
$(\ell_1, r_1)^Y$ & $(\ell_2, r_2)^Z$ & $(\ell_3, r_3)^X$\\
$\cdots$ & $\cdots$ & $\cdots  $
\end{tabular}
\end{center}
The table depicts the bundles that each agent regards as at least as good as their endowment, and exhibits three allocations by means of superscripts. For example, the allocation $\{X_1,X_2,X_3\} =\{(\ell_1, r_2),(\ell_2, r_1),(\ell_3, r_3)\}$ results from agents 1 and 2 trading right shoes, while agent 3 consumes her endowment.
\footnote{Trading of right and left shoes might suggest a version of the TTC for each type of shoe. Example~\ref{ex:shoes} shows that this approach does not work. It does not produce a core allocation, even for economies with additively-separable preferences over shoes (\cite{konishi2001shapley, biro2022serial, feng2022preference}).}

Now it is easy to verify that the weak core is empty. The allocation $X$ is blocked by agents 2 and 3, using the allocation $Y$ (in fact, using a $\{2,3\}$-allocation). The allocation $Y$ is blocked by agents 1 and 3 using the allocation $Z$, which is in turn blocked by agents 1 and 2 by means of allocation $X$.
\end{example}

Some of our results will involve limits on the size of the possible blocking coalitions. An allocation is \df{pairwise stable} if it is individually rational and not strongly blocked by any coalition of size two. The motivation for considering such laxer stability concepts is that it may be difficult for agents to coordinate blocks among large coalitions. Pairwise stability is the most basic collective bargaining model, focusing on the smallest possible non-trivial blocking coalitions.  

When we want to talk loosely, and informally, about allocations that are robust to blocking by some family of coalitions, we shall simply call them \df{stable}.

The problem we are facing is, in a sense, harder than the roommate problem, for which stable allocations are known not to exist. We may interpret the trades between pairs of agents in Example~\ref{ex:shoes} as pairs in the roommate problem. No configuration of pairwise trades, translated into a configuration of couples and singles, is stable. For example, if agent 1 pairs up with 2, and leaves agent 3 ``single,'' (consuming their endowment) then 2 and 3 can block by pairing up instead. Stability in the roommate problem is famously challenging, and the logic behind the non-existence of stable outcomes in the roommate problem \emph{is the same} as in our ``shoe economy'' example. This means that the model of a discrete exchange economy is rich enough to encode some of the best-known examples of instability in game theory.

\cite{ehlers2007neumann}, \cite{mauleon2011neumann}, and  \cite{ehlers2020legal} consider versions of the von-Neumann-Morgenstern stable sets in assignment problems. Stable sets are also an area of interest for the discrete economies that we consider in our paper. Example~\ref{ex:shoes} shows that the stable set can be empty. If we focus on individually rational allocations, we see that the endowments and three alternative allocations $\{X,Y,Z\}$ should be considered. The three allocations are cyclically blocked by each other: $X$ is blocked by $Y$, which is blocked by $Z$, and in turn by $X$. If a stable set includes allocation $X$, then it must include the allocation $Z$ to satisfy external stability. This, however, violates internal stability. So one would have to impose additional assumptions in order to obtain the existence of a stable set.

\section{Core in a categorical economy}
\label{sec:results}

In light of our previous discussion and examples, it seems impossible to obtain a general existence result for the core. The model is arguably rich enough to replicate any problematic behavior that can be exhibited in a NTU game --- such as the non-existence of stable outcomes in the roommate problem. The situation is, however, not hopeless. There is space to add structure in ways that can ensure existence. We propose sufficient conditions on the primitives of an economy that allow us to prove the existence of stable outcomes. 

Towards a formulation of sufficient conditions, it is useful to recall the arguments
that \cite{shapley1974cores} use for the housing model. Shapley and Scarf use two very different arguments to show that the weak core is nonempty in the housing model: Gale's TTC algorithm and Scarf's
theorem on the core of NTU games. The TTC obtains a core allocation constructively for the housing model, but extending the algorithm to a more general environment poses a challenge, because the algorithm only allows for the trading of the same number of items among agents at a time. In a general exchange economy, a core allocation might involve an agent giving one item to another while receiving multiple items from different agents. We turn to a generalization of TTC in Section~\ref{sec:Talgorithm}. Instead, the focus in the current section is on sufficient conditions that allow us to use Scarf's theorem to show the existence of nonempty weak core.  To this end, we first review Scarf's theorem and discuss its application to the housing model.

\subsection{Scarf's theorem.}\label{sec:scarfsthm}

Given an economy $E = (O, \{(v_i, \w_i): i\in A\})$, we define a NTU game $(A,V)$ by defining the set $V(S)\subseteq\Re^{A}$ for each nonempty $S \subseteq A$ to be the set of all vectors $u\in\Re^A$ for which there exists an $S$-allocation $\{X_i : i \in S\}$ with $u_i\leq v_i(X_i)$ for $i \in S$. Lastly, we let  $V(\emptyset) = \{0\}$.


A family of subsets of $A$, denoted by $\S$, is said to be a \df{balanced collection of coalitions} if there exist non-negative weights $(\delta_S)_{S \in \S}$  such that
$$\sum_{\{S\in \S : S\ni i \}} \da_S=1 \text{ for all } i \in A.$$
A NTU game $(A,V)$ is \df{balanced} if for every balanced collection of coalitions $\S$ and $u \in \Re^A$,  $$u \in \cap_{S \in \S} V(S) \implies u \in V(A).$$

\begin{lemma}[\cite{scarf1967core}]\label{lem:scarf}
A balanced NTU game has a nonempty weak core.
\end{lemma}

\label{page:editorscarf}
\cite{shapley1974cores} show that the NTU game defined by the housing model is balanced, and therefore has a nonempty weak core. Their proof consists of two key steps: first, they establish a \df{fractional} allocation that attains  ``target utilities.'' Then, in the second step, they \df{round} the fractional allocation to obtain an integral allocation, while preserving the agents' target utilities.

It helps to discuss this proof approach in more detail, as we propose a way to generalize it to a multi-item model. Let $\S$ be a balanced collection of coalitions with weights $(\da_S)_{S\in \S}$. We regard any $u\in\cap_{S \in \S} V(S)$ as target utilities. Notice that each $S \in \S$ has an $S$-allocation $\{X_i : i\in S\}$ such that $v_i(X_i) \geq u_i$ for $i \in S$. So the target utilities are achieved in trades within the coalitions that make up the balanced collection. 

Given such coalition-specific allocations, we may consider a candidate economy-wide allocation by using weights $(\da_S)_{S\in \S}$ and calculating a weighted average. Let $P_S$ be a $\abs{A} \times \abs{O}$ matrix describing the $S$-allocation $\{X_i : i\in S\}$ so that $P_S(i,j)=1$ if agent $i$ receives object $j$ and $0$ otherwise (for simplicity, suppose that each agent in $S$ receives one object in the $S$-allocation). Then,
\[
P = \sum_{S\in \S} \da_S P_S
\] is a ``fractional'' economy-wide allocation.\footnote{\label{fn:fractional_matrix}Here is a proof: An agent $i$ receives an object if $i \in S$, hence $\sum_j P (i,j) = \sum_{S \in \S; S \ni i} \delta_S (\sum_j P_S(i,j)) = 1$. Similarly, an object $j$ is allocated to an agent when the owner of $j$, say $i'$, is in $S$, so $\sum_{i} P(i, j) = \sum_{S \in \S; S \ni i'} \delta_S (\sum_{i} P_S(i, j)) = 1$.} If $P(i,j) > 0$, then $P_S(i,j) = 1$ for some $S \in \S$, so agent $i$ obtains at least the target utility $u_i$ from object~$j$. 


We may interpret $P$ as \textit{random allocation}: by the Birkhoff-von Neumann theorem, the matrix $P$ can be written as a lottery over assignments of objects to agents. In a model that allows random allocations, and where agents have expected-utility preferences for lotteries over objects, each $P_S$ could represent a random allocation for coalition $S$, ensuring the target expected utility $u_i$ for agent $i \in S$. Then $P$ represents a random allocation ensuring the target expected utility $u_i$ for agent $i \in A$, establishing $u \in V(A)$. According to \cite{shapley1974cores}, the first step would already establish that the NTU game is balanced. There is a  nonempty weak random core.

In a discrete allocation problem, however, random allocations are not allowed. We must \textit{round} the fractional matrix $P$ to obtain a deterministic allocation that ensures that agents obtain their target utilities. This is the second step in \cite{shapley1974cores}. Rounding is easy for the housing model. According to the Birkhoff–von Neumann Theorem, $P$ can be expressed as a weighted average of integral matrices, and by inspecting the proof of this theorem one can see that target utilities are preserved.  Thus, there is a deterministic allocation with a matrix $\overline{P}$ such that, if  $\overline{P}(i, j) = 1$, then $P (i,j) > 0$, which implies that agents attain their target utilities. Thus, $u \in V(A)$; so the game is balanced, and the weak core is nonempty by Lemma~\ref{lem:scarf}.

\subsection{Additively separable preferences}
A natural extension of the housing model is the categorical economy (a
``house-car-boat economy''). An economy $E=(O,\{(v_i,\w_i):i\in A \})$ is
\df{categorical} if there exists a number $K$ of \df{categories of objects} and sets
of objects $O^k$, for each category $k\in \{1\ldots  K\}$ so that agents has a unit
demand in each category. There are no restrictions on the endowments (for example,
one agent could be endowed with all objects in one category).

As mentioned in the introduction, \citet{moulin2014cooperative} proposes the model of
a categorical economy, and leaves open the question of whether the core is empty. But
a left-right shoe economy is a special case of a categorical economy, and we observed in
Example~\ref{ex:shoes} that the weak core in a shoe economy may be empty.  Our goal is to find conditions on preferences such that the NTU game defined by a categorical economy is balanced, and therefore has nonempty weak core.

Recall our discussion of \cite{shapley1974cores}. We shall mimic the approach that they pursue. First we consider the possibility of a nonempty weak random core. It turns out that additively separable preferences ensure that there is a random core allocation. Then we turn to a rounding algorithm that will produce a non-random core allocation out of the random core allocation. This second step, we shall see, requires further assumptions on agents' preferences.

For nonempty $S \subseteq A$, a random $S$-allocation $\sigma_S$ is a probability
distribution (a lottery) over $S$-allocations. We assume that agents have a von-Neumann-Morgenstern expected utility, which we also denote by $v_i(\sigma_S)$. A coalition $S$ strongly blocks a random allocation $\sigma$ if there exists a random $S$-allocation $\sigma_S$ such that $v_i(\sigma_S) > v_i(\sigma)$ for all $i \in S$. The weak random core is the set of random allocations that are not strongly blocked by any coalition $S$.

The random core of a categorical economy may still be empty. In Example~\ref{ex:shoes}, with two categories (left shoes and right shoes), suppose that each agent's vNM utilities are 3, 2, and 0 for the three acceptable pairs, respectively, and that they receive a utility of $-\infty$ for any unacceptable pair. A random allocation with a probability distribution $(p_X, p_Y, p_Z)$ over the three individually rational allocations must satisfy the following inequalities. For agent 1: $p_X \cdot 3 + p_Z \cdot 2 + p_Y \cdot 0 \geq 2$, for agent 2: $p_Y \cdot 3 + p_X \cdot 2 + p_Z \cdot 0 \geq 2$, and for agent 3: $p_Z \cdot 3 + p_Y \cdot 2 + p_X \cdot 0 \geq 2$. Adding up both sides of these inequalities results in $5 \geq 6$. So the random core is empty.

Given the message of the shoe economy in Example~\ref{ex:shoes}, we must further strengthen the
assumptions we have placed on preferences. We consider utilities that are additively separable over categories. Each agent $i$ is endowed with a utility $v_i^k: O^k \to \Re$ for each category $k$ such that $v^k_i(\emptyset) = 0$; and $v_i(X_i) = \sum_{k} v^k_i(X_i \cap O^k)$. Here we adopt the convention that $v^k_i(X_i \cap O^k)= \max_{o \in X_i \cap O^k} v^k_i(o)$ if there is some object $o \in X_i \cap O^k$ with $v^k_i(o) > 0$, and $0$ otherwise.

\begin{proposition}
\label{prop:fractional_core}
A categorical economy with additively separable utilities has nonempty weak random core.
\end{proposition}

\begin{proof}
The proof is an application of Lemma~\ref{lem:scarf}. Consider the NTU game $(A, V)$ such that for nonempty $S \subseteq A$, $V(S) \subseteq \Re^A$ is the set of $u \in \Re^A$ for which there exists a random $S$-allocation $\sigma_S$ with $u_i \leq v_i(\sigma_S)$ for $i \in S$. 

Let $\S$ be a balanced collection of coalitions with weights $(\da_S)_{S\in \S}$ so that $\sum_{\{S\in \S : S\ni i \}} \da_S=1$. Take any $u \in \cap_{S \in \S} V(S)$. For each $S \in \S$, there is a random $S$-allocation $\sigma_S$ such that $v_i(\sigma_S) \geq u_i$ for $i \in S$. The random $S$-allocation $\sigma_S$ defines a $\abs{A} \times \abs{O}$ fractional matrix $P_S$ such that $P_S(i,j)$ is the probability that object $j$ is allocated to agent $i$. We may write $P_S$ by collecting the columns that correspond to goods of the same category: $P_S=[P_S^1 | P_S^2 |\dots|P_S^K]$, with each $P_S^k$ being a fractional allocation of objects of category $k$. 

Finally, $P = \sum_{S \in \S} \delta_S P_S$. Observe that 
$P=[P^1 | P^2 |\dots|P^K]$, where each $P^k = \sum_{S \in \S} \delta_S P^k_S$. Just as in the proof of Lemma~\ref{lem:scarf}, we may apply the Birkhoff-von Neumann theorem and interpret each $P^k$ as a random allocation, a lottery over economy-wide assignments of category-$k$ objects to agents.  Then, $P$ can be written as the product of the lotteries $P^1, P^2, \dots, P^K$.

Since agents' vNM utilities are additively separable across categories,
\begin{align*}
v_i(P)
& =  \sum_k \left(\sum_{o_j \in O^k} v_i(o_j) P^k (i,j) \right)   = \sum_k \sum_{o_j \in O^k} v_i(o_j)  \sum_{\{S \in \S: S \ni i\}} \delta_S  P^k_S (i,j) \\
& = \sum_{\{S \in \S: S \ni i\}} \delta_S \left( \sum_k \sum_{o_j \in O^k} v_i(o_j) P^k_S (i,j) \right) = \sum_{\{S \in \S: S \ni i\}} \delta_S v_i(\sigma_S) \\ 
&\geq u_i.
\end{align*}
Hence, $u \in  V(A)$. Thus the game $(A, V)$ is balanced, and the weak random core is nonempty.
\end{proof}

\subsection{Dichotomous preferences.}

A categorical economy with additively separable preferences has a nonempty random core. Unfortunately, the (deterministic) weak core can still be empty. We proceed with a counterexample. 

\begin{example}\label{ex:KQW23} Consider 
an economy with $K=2$ categories, in fact a left-right shoe economy. Suppose that there are four agents with ordinal preferences given as follows:
$$
\begin{array}{l}
(l_4,r_1) \succ_{1}(l_4,r_2) \succ_{1}(l_1,r_1) \succ_{1} ...\\
(l_2,r_1) \succ_{2}(l_1,r_1) \succ_{2}(l_2,r_3) \succ_{2}(l_2,r_2) \succ_{2} ...\\
(l_3,r_4) \succ_{3}(l_3,r_2) \succ_{3}(l_4,r_4) \succ_{3}(l_3,r_3) \succ_{3} ...\\
(l_3,r_4) \succ_{4}(l_3,r_3) \succ_{4}(l_2,r_4) \succ_{4}(l_4,r_4) \succ_{4} ...\\
\end{array}
$$

This example coincides with Example 2.3 in \cite{konishi2001shapley}, who show that the weak core in this game is empty. We shall not prove this fact here, and instead refer the reader to Konishi et al.\ for a proof. The point we do wish to make is that the example is consistent with additively separable preferences. 

The next table provides an additively separable utility representation:
\begin{center}
\begin{tabular}{c|cccc}
$(v^l_i, v^r_i)$ & Agent 1 & Agent 2    & Agent 3  & Agent 4 \\\hline
 $(l_1, r_1)$ & $1, 3$ & $3, 5$ & $0, 0$ & $0, 0$ \\
$(l_2, r_2)$ & $0 , 2$ & $5, 1$ & $0, 3$ & $2,0$\\
 $(l_3, r_3)$ & $0, 0$ & $0, 2$ & $5, 1$ & $5, 3$\\
 $(l_4, r_4)$ & $3, 0$ & $0, 0$ & $2, 5$ & $1, 5$\\
\end{tabular}
\end{center}
For instance, agent $1$'s ordinal preference $(l_4, r_1) \succ_1 (l_4, r_2)$ is represented by $v^l_1(l_4) + v^r_1(r_1) = 3 + 3 > v^l_1(l_4) + v^r_1(r_2) = 3 + 2$.
\label{ex:add_separable}
\end{example}

Given Example~\ref{ex:KQW23}, there is no hope of obtaining a nonempty core in the
categorical economy, even with additively separable preferences. The random core is
nonempty, but the rounding step in \cite{shapley1974cores}, allowing us to turn a random core allocation into a deterministic core allocation, does not generalize. So we turn to a further restriction on agents' preferences: we consider dichotomous preferences. 

In the context of a categorical economy with additively separable utilities, where $v_i(X_i) = \sum_{k} v^k_i(X_i \cap O^k)$, dichotomous preferences mean that for each agent $i$, there is a set of objects $G_i$ such that $v^k_i(X_i \cap O^k) = 1$ if there is some object $o \in X_i \cap O^k \cap G_i$, and $0$ otherwise. Each agent classifies items as either ``acceptable'' or ``unacceptable'' goods and evaluates a bundle based on the number of acceptable goods it contains, counting at most one in each category.

A special case of our model is a ``non-categorical'' economy where each object forms a separate category. Then an agent $i$'s utility from a bundle equals the number of acceptable goods it contains: $v_i(X_i) = \abs{X_i \cap G_i}$. Models of markets and social choice with dichotomous preferences have been studied quite extensively; see, for example, \cite*{bogomolnaia2004random}, \cite*{bogomolnaia2005collective}, and \cite*{aziz2019fair}. 

\begin{theorem}\label{thm:categorical_economy} A categorical economy with additively separable and dichotomous preferences has nonempty weak core.
\end{theorem}

The proof can be found in Section~\ref{sec:proofs}. Here we proceed to describe the main ideas in the proof, and to explain how it follows along the lines of Shapley and Scarf's strategy for proving nonemptyness of the core in the housing model.

Proposition~\ref{prop:fractional_core} shows that the random core is not empty. Consider then a random core allocation $\sigma$, and let $t_i$ be the floor of (the largest integer less than or equal to) the expected utility $v_i(\sigma)$. We treat these integers $(t_i)_{i\in A}$ as \emph{target utilities}. The main idea in the proof is to round $\sigma$ so as to obtain a deterministic allocation $X$ such that $v_i(X_i) \geq t_i$ for $i \in A$.  The allocation $X$ must be in the (deterministic) weak core: If $X$ is blocked by a coalition $S$ with a deterministic $S$-allocation $(Y_i)_{i\in S}$, then $v_i(Y_i) > v_i(X_i) \geq t_i$ for $i \in S$. The utilities $v_i(Y_i)$ must be integers, due to the assumption of dichotomous preferences. Hence  $v_i(Y_i) \geq t_i + 1 > v_i(\sigma)$ for $i \in S$, contradicting the assumption that $\sigma$ is a random core allocation.

\subsection{The rounding argument.}\label{sketch:rounding}

Given our discussion so far, it should be clear that they key new idea in our result is the rounding
argument. Rounding here is much more involved than the argument in \cite{shapley1974cores}
for the housing model. We explain the proof idea for the special case in which each object forms a separate category ($\abs{O^k}=1$ for $k=1, \dots, K$), so $v_i(X_i) = \abs{X_i \cap G_i}$ for $i\in A$. See Section~\ref{sec:proofs} for a description of the rounding algorithm in the general case. 


Consider a random core allocation $\sigma$, and let $t_i$ be the floor of the expected utility
$v_i(\sigma)$. Without loss of generality, assume that target utilities satisfy $t_i\geq 1$ (the
case when $t_i=0$ may be dealt with trivially).
The random allocation $\sigma$ defines a $\abs{A} \times \abs{O}$ matrix $P$ such that $P(i,j)$ is the probability that object $j$ is allocated to agent $i$. Observe that $P$ has two properties:
\begin{enumerate}
    \item Each row $i$ sums up to at least $t_i\geq 1$.
    \item Each column $j$ sums up to at most $1$.
\end{enumerate}
We round $P$ to obtain a deterministic allocation $X$ such that $v_i(X_i) \geq t_i$ for $i \in A$. 


Let $B$ be any $\abs{A} \times \abs{O}$ matrix that satisfies these two properties. Say that an entry $b_{i,j}$ is fractional if $b_{i,j}\in (0,1)$. 

A one is an entry $b_{i,j}=1$. Suppose that there is a row $i$ that has less than $t_i$ entries that are ones. Since row $i$ adds up to at least $t_i$, and no entry can be larger than one, there must exist a column $j$ for which the entry $b_{i,j}\in (0,1)$ is fractional. Now there are two possibilities. Either $b_{i,j}$ is the only non-zero entry in column $j$, or there is another row $h$ with $b_{h,j}>0$. If the former case is true, we can increase $b_{i,j}$ to one and still obtain a matrix that satisfies the two properties but has strictly fewer fractional entries.

Suppose then that the latter case is true. Since $b_{i,j}>0$ and $B$ satisfies property (2), $b_{h,j}<1$. Now there  are again two cases. Either there are $t_h$ ones in row $h$, or not. If there are already $t_h$ ones in row $h$ then, since $b_{h,j}$ is fractional, the columns with ones cannot include $j$. So we can decrease $b_{h,j}$ without affecting the property that the entries of row $h$ sum up to at least $t_h$. We may then define the largest $\ep>0$ such that $b_{i,j}+\ep\leq 1$ and $b_{h,j}-\ep\geq 0$. Then replace $b_{i,j}$ with $b_{i,j}+\ep$ and $b_{h,j}$ with $b_{h,j}-\ep\geq 0$. We obtain a new matrix that satisfies the two properties but has strictly fewer fractional entries.

Finally, suppose that there are not $t_h$ ones in row $h$. Given that row $h$ adds up to at least $t_h$, there must exist another fractional entry $b_{h,l}\in (0,1)$.  Starting from  $b_{h,l}\in (0,1)$, repeat the argument above.

This procedure will either result in a new matrix that satisfies the two properties and has strictly fewer fractional entries, or we will obtain a cycle $(i_1,j_1),\ldots,(i_M,j_M)=(i_1,j_1)$ where \begin{itemize}
    \item $b_{i_m,j_m}\in (0,1)$ is fractional;
    \item each odd-numbered entry is in the same row as the next entry, and in the same column as the preceding one.
\end{itemize}
 Now we may add $\ep>0$ to each odd-numbered entry, and subtract $\ep$ from the following entry in the cycle. This transformation will preserve the same row and column sums, as we add and substract $\ep$ the same number of time for each row and column. Choose the largest $\ep$ that ensures that each entry is in $[0,1]$. This again results in a new matrix that satisfies properties (1) and (2) and has strictly fewer fractional entries.

The argument we laid out implies that for any matrix satisfying (1) and (2), as long as there is some row $i$ that does not have $t_i$ 
ones, there is a new matrix that satisfies properties (1) and (2) and has strictly fewer fractional entries. So there must then exist a matrix satisfying the two properties and where all rows have at least $t_i$ ones. Finally, obtaining a deterministic allocation that achieves the target utilities from this matrix is straightforward.

See Section~\ref{sec:proofs} for the general rounding argument that we use to prove the theorem.

\section{A generalization of the TTC algorithm}\label{sec:Talgorithm}

We propose an algorithm for the problem of finding stable allocations in discrete exchange economies. Aside from desirable computational properties, the algorithm has two advantages: When applied to the standard Shapley-Scarf model of a housing market, it replicates the TTC and finds a weak core allocation (Section~\ref{sec:TTTC}). In general, the same algorithm finds, under certain conditions, a set of allocations that is stable in a farsighted sense (Sections~\ref{sec:BSresults} and~\ref{sec:replica}). More specifically, the algorithm constructs a set of allocations with the property that any blocking coalition would be subject to a sequence of forward-looking (``farsighted'') blocks that would end with another allocation in the set.

We impose additional conditions on utility functions. A utility function $v_i$ is \df{injective} if $X_i\neq X'_i$ implies $v_i(X_i)\neq v_i(X'_i)$; and it is (strictly) \df{monotone} if $X_i\subsetneq X'_i$ implies $v_i(X_i) < v_i(X'_i)$. The property of injectivity corresponds to strict preferences, and is very common in matching theory and the discrete assignment literature. Allowing for indifferences presents problems, even in the simple housing market (\cite{wako1991some} shows that the strong core may be empty in the Shapley-Scarf housing model when agents can be indifferent between different houses). The monotonicity assumption as stated is actually stronger than needed. For example, for our purposes, we may think of the housing model as having monotone utilities, even though it fails the definition given above.  The key assumption is monotonicity on the appropriate domain of agents' utilities. The housing model's unit demand assumption violates the definition of monotonicity, but such violation is immaterial for our purpose as each agent obtains exactly one house in any individually-rational exchange. For simplicity, however, we assume monotonicity on the full domain.

Let $U_i=v_i(2^O)$ be the set of all possible values that $i$'s utility can take. Since $v_i$ is injective, $U_i$ is a finite set with cardinality $2^O$; in fact, each utility value $u_i\in U_i$ may be identified with a bundle $v^{-1}_i(u_i)\subseteq O$. 

We call a collection $X=\{X_i : i \in A\}$ with $\cup_i X_i \subseteq O$ a \df{preallocation}. A pairwise-disjoint preallocation is an allocation. A preallocation $X$ is \df{individually rational} if $v_i(X_i)\geq v_i(\w_i)$ for all $i$. A profile of utilities $u=(u_1,\ldots,u_n)\in U=\times_{i\in A} U_i$ is identified with a preallocation.

Given $u\in U$ and $k=1, \ldots, n$, let 
\[\begin{split}
    B^k_i(u)=\left\{u_i': \exists S\text{-allocation } X \text{ st } |S| \leq k, i \in S,  v_i(X_i)=u_i' \right. \\
    \left. \text{ and for all } j \in S \setminus \{i\};   v_j(X_j) \geq u_j \right\}. 
\end{split}\]
$B^k_i(u)$ denotes the set of utilities that agent $i$ can achieve through a trade in a coalition of size at most $k$, and that assures each of $i$'s trading partner $j$ at least the utility $u_j$. Since $i$ can form a coalition $\{i\}$, $v_i(\w_i) \in B^k_i(u)$.

Now define a function $T_k:U\to U$ by $(T_k u)_i=\max B^k_i(u)$. We use the notation $T_k u$ rather than $T_k(u)$ for the values of the function $T_k$. In words, the function $T_ku$ lets each agent $i$ obtain the best possible utility that she can achieve by forming a trade with at most $k-1$ other agents, and by ensuring that each of her trading partners enjoys at least the utility that they are guaranteed to obtain in the vector $u$. Importantly, the vector $u$ corresponds to a \emph{pre}-allocation: it could be a ``fantasy'' in which multiple agents consume the same goods. We shall see in Lemma~\ref{lem:Talgobasic}, however, that when there exist no indifferences due to injective utility functions, any fixed point of $T_k$ corresponds to an actual allocation.

In the following, 
we shall suppress the dependence of $T$ on $k$ in our notation. The algorithm we are interested in consists of iterates of the function $T$, and so we shall denote by $T^m$ the composition of $T$ with itself $m$ times. Formally, the $m$-th iterate of $T$ is defined recursively  by $T^m u = T(T^{m-1}u)$ for $m=2, 3, \dots$. The key observation is that the function $T$ is monotone decreasing, so that $u\leq u'$ implies that $Tu'\leq Tu$. In turn, this means that the composition of $T$ with itself, $T^2$, is monotone increasing. Tarski's fixed point theorem therefore applies to the mapping $u\mapsto T^2u$.

\begin{remark}
For the left-right shoe economy in Example~\ref{ex:shoes}, the $T$-algorithm works as follows. Let $\{0, 1, 2\}$ be the utility representations of each agent's three most preferred left-right shoe pairs. In particular, $u=(0,0,0)$ is the agents' utilities from consuming their endowments: $v_i(l_i, r_i) = 0$ for $i=1,2,3$. Any agent $i$ can find another agent, say $j$, such that a pairwise exchange yields the utility of $2$ for $i$, while agent $j$ obtains the utility of $1 (\geq 0)$. It follows that $Tu=(2,2,2)$, and $T(2,2,2) = T^2 u = (1,1,1)$. A further applications of $T$ obtain $T^m u =(1,1,1)$ if $m$ is an even number, and $T^m u=(2,2,2)$ if $m$ is an odd number. Thus, $(1,1,1)$ is a fixed point of $T^2$ and $(1,1,1) \leq T (1,1,1) = (2,2,2)$.
\end{remark}

\subsection{Gale's TTC algorithm}\label{sec:TTTC}

When applied to the housing market in \cite{shapley1974cores}, our algorithm mimics the TTC, \textit{iteration by iteration}.

Specifically, we consider a housing market $\{ A, O = (h_i)_{i \in A}, (\succ_i)_{i \in A}\}$. Let $v_i$ be a utility representation of $i$'s preferences $\succ_i$.  Let $\underline{u} = (v_i(h_i))_{i \in A}$ be defined by the allocation in which each agent is consuming their own house. 

We consider the function $T:U \to U$ defined above, allowing for coalitions of any size ($k=n$). An iteration of $T$ starts from the vector $\underline{u}\in U$: $T^m\underline{u}$, $m=1, 2, \dots$. Observe that $\underline{u} \leq T\underline{u}$. Since $T$ is decreasing and $\underline{u}$ is the minimum individually rational pre-allocation, $\underline{u} \leq T^2\underline{u} \leq T\underline{u}$. Then, as $T$ is decreasing, $T\underline{u} \geq  T^3\underline{u} \geq T^2\underline{u}$. Subsequently, 
\begin{equation}
    \label{eqn:semi-monotonic-T}
\underline{u} \leq T^2\underline{u}\leq T^4 \underline{u} \leq \dots \leq T^3\underline{u} \leq T\underline{u}.
\end{equation}

We consider a version of TTC that clears all cycles in each round. We index each round of TTC by $r$, and denote by $A_r\subseteq A$ the set of agents who obtain their final house in the $r$-th round of TTC. The next result states that the \emph{even} iterations of $T$ mimic the different rounds of the TTC algorithm.

\begin{lemma}\label{lem:TTC_extension}
$\forall r=1,2, \dots$ and $i \in A_r$, 
$$(T^{2r-1}\underline{u})_i= (T^{2r}\underline{u})_i= (T^{2r+1}\underline{u})_i=\dots.$$
\end{lemma}    
In words, agent $i$ updates their expectation (fantasy) during the first $2r-1$ iterations of $T$ and learns that the last update is achievable after the $2r$-th iteration. Furthermore, the utility profile from the TTC allocation is a fixed point of $T$. 

\begin{remark}
While there exists a literature on extending TTC, the primary focus has been on achieving a strategy-proof mechanism with additional properties, such as Pareto efficiency and individual rationality. \cite{alcalde2011exchange}, \cite{jaramillo2012difference}, \cite{aziz2012housing}, and \cite{saban2014note} have extended the TTC algorithm to a housing model allowing for indifferences. In addition, \cite{papai2003strategyproof}, \cite{papai2007exchange}, \cite{todo2014strategyproof}, and \cite{fujita2018complexity} consider trade mechanisms in a general discrete exchange economy. Our paper does not aim to obtain a strategy-proof mechanism. 

Instead, we extend the TTC algorithm to obtain allocations that satisfy some notion of stability.
\end{remark}

\subsection{Fixed points of $T$ and $T^2$}

The formulation of the function $T$ is inspired by similar constructions in the literature on two-sided matchings.\footnote{See, for example, \cite{adachi2000characterization}, \cite{ECHENIQUEA2004358}, \cite{echenique2004theory}, \cite{fleiner2003fixed}, \cite{hatfield2005matching}, and \cite{ostrovsky2008stability}.} There the two-sided nature of the problem allows for the existence of a fixed point of $T$, under standard assumptions on preferences. Trading in our model is \emph{not two sided}; so there is no hope of replicating the ideas in the two-sided models. Instead, we use an idea exploited by \cite{roth1975lattice}, \cite{echenique2007solution}, and \cite{ehlers2020legal} to work with the fixed points of $T^2$.

The following result establishes the basic properties of the fixed points of $T$ and of $T^2$. 

\begin{lemma}\label{lem:Talgobasic} Consider the function $T:U\to U$ as defined above (for an arbitrary value of $k$).
\begin{itemize}
    \item There exists $u\in U$ such that $u=T^2 u$, $u \leq Tu$, and the preallocation defined by $u$ is individually rational.
    \item If $u=T u$, then the preallocation defined by $u$, namely $\{v^{-1}_i(u_i) : i \in A\}$, is an allocation that is individually rational and not blocked by any coalition of size at most $k$.
\end{itemize}
\end{lemma}

\begin{remark}
Conversely, if there is no restriction on coalition sizes (meaning $k=n$), then $u$ defined by any core  allocation satisfies $u = Tu$. The proof is simple. Let $X$ be a core allocation and $v_i(X_i) = u_i$. Then, $u_i \in B_i^k(u)$ by definition of $X$, so $u_i \leq (Tu)_i$. However, $u_i < (Tu)_i$ contradicts the individual rationality, or the no blocking of the allocation $X$, so $u_i = (Tu)_i$. In general, if $k < n$, then $u$ defined by a core allocation may not satisfy $u=Tu$; the core allocation may require trade between many agents.
\end{remark}

A fixed point $u$ of $T^2$ such that $u \leq Tu$ exists by Tarski's fixed point theorem, and the set of fixed points of $T^2$ is guaranteed to contain all fixed points of $T$ ($u=Tu$ implies $u=T^2 u$). A fixed point of $T$, however, may not exist. Then, a fixed point of $T^2$ defines an allocation only for the coalition $\{i : u_i = (Tu)_i\}$.

\begin{lemma}\label{lem:subsetofagents} 
Take $u$ such that $u = T^2 u$ and $u \leq T u$, and define $A_1 = \{ i : u_i = (Tu)_i\}$. Then, the preallocation defined by $\{ v_i^{-1}(u_i) : i \in A\}$ is an $A_1$-allocation: (1) $\cup_{i \in A_1} v^{-1}_i(u_i) = \cup_{i \in A_1} w_i$ and (2) $v^{-1}_i(u_i) \cap v^{-1}_{i'} (u_{i'}) = \emptyset$ for $i,i' \in A_1$.
\end{lemma}

\begin{remark}
An algorithm exists that, with oracle access to $T$, finds a fixed point of $T^2$ in time $O(\log^{\abs{A}} M)$, where $M=\max\{\abs{\{X_i : v_i(X_i)\geq v_i(\w_i) \}}:i\in A \}$. This algorithm combines the standard idea of iterating the monotone increasing function $u\mapsto T^2 u$ on the lattice $U$ with binary search and is easily formulated given the recent literature on finding a Tarski fixed point: see, for example, \cite{etessami_et_al:LIPIcs:2020:11703}.\footnote{The recent literature on finding a Tarski fixed point includes \cite{dang2020computations} (in unpublished work that was perhaps the first in proposing the combination of iteration and binary search), \cite{chang2008complexity}, \cite{dang2018complexity}, and \cite{fearnleyfaster}.}
\end{remark}

\subsection{Application ($k=2$)}

A fixed point $u$ of $T^2$ always exists. Lemma~\ref{lem:Talgobasic} shows that if $u$ is also a fixed point of $T$, then there exists an allocation $X$ with $v_i(X_i)=u_i=(Tu)_i$, and $X$ is unblocked (by any coalition of size up to $k$) . However, a fixed point $u$ of $T^2$ may not be a fixed point of $T$. So, we consider a weaker solution concept than the core. We consider whether there exists an allocation $X$ such that $v_i(X_i)$ equals to either $u_i$ or $(Tu)_i$ for $i \in A$, and explore whether the set of such allocations satisfies a stability property.

We will consider the function $T: U \to U$ when $k=2$, and a fixed point $u=T^2u$ with $u \leq Tu$. For any allocation $X$ where $v_i(X_i) \geq u_i$ for all $i$, the gap $Tu - u \geq 0$ serves as an upper bound on instability. If there exists a pairwise block of $X$ by an $\{i,j\}$-allocation $(Y_i, Y_j)$, the utility gains are at most $(Tu)_i - u_i$ and $(Tu)_j - u_j$. 


Here is an overview of our approach. For agents in $A_1 = \{ i : u_i = (Tu)_i\}$, there exists an $A_1$-allocation defined by $u$ (Lemma~\ref{lem:subsetofagents}). For any other agent $i \in A_2=\{i' : u_{i'} < (Tu)_{i'}\}$, the utilities $u_i$ and $(Tu)_i$ are obtained either from the consumption of $i$'s endowment, or from a bundle obtained by a pairwise exchange with another agent in $A_2$. Thus, each agent $i \in A_2$ has a unique ``trading partner'' $j \in A_2$ such that $((Tu)_i, u_j)$ defines an $\{i,j\}$-allocation (Lemma~\ref{lem:cycle_construction}). Subsequently, the agent $j$ has a unique trading partner $k \in A_2$ such that $((Tu)_j, u_k)$ defines an $\{j,k\}$-allocation. It follows that these agents are in a cycle $i \to j \to \dots \to i$; and such cycles partition $A_2$.

For an illustration, suppose that an economy has 4 agents $A=\{i_1,i_2,i_3,i_4\}$. Consider a fixed point of $u=T^2u$ with $u \leq Tu$, and in which all agents are in the same cycle $i_1 \to i_2 \to i_3 \to i_4 \to i_1$, as defined above. Then, the allocation $Y$ defined by $((Tu)_{i_1}, u_{i_2}, (Tu)_{i_3}, u_{i_4})$, and the allocation $Z$ defined by $((Tu)_{i_2}, u_{i_3}, (Tu)_{i_4}, u_{i_1})$ are such that every agent $i$ obtains either $u_i$ or $(Tu)_i$. (Note that we use injectivity of utility to talk about allocations in goods or utility space interchangeably.)

\begin{figure}[h!]
\centering
\begin{tikzpicture}[->,>=stealth,shorten >=1pt,auto,node distance=2.2cm,
                    thick,main node/.style={circle,draw,font=\sffamily\bfseries, inner xsep=-1pt, inner ysep=-1pt}]

  \node[main node,draw=none] (1) {$(Tu)_{i_1}$};
  \node[main node,draw=none] (2) [right of=1] {$(Tu)_{i_2}$};
  \node[main node,draw=none] (3) [right of=2] {$(Tu)_{i_3}$};
  \node[main node,draw=none] (4) [right of=3] {$(Tu)_{i_4}$};
  \node[main node,draw=none] (5) [below of=1, node distance=2.5cm] {$u_{i_1}$};
  \node[main node,draw=none] (6) [below of=2,node distance=2.5cm] {$u_{i_2}$};
  \node[main node,draw=none] (7) [below of=3,node distance=2.5cm] {$u_{i_3}$};
  \node[main node,draw=none] (8) [below of=4,node distance=2.5cm] {$u_{i_4}$};

    \node[ellipse, draw, inner xsep=2pt, inner ysep=-10pt, fit=(1) (5), label={$i_1$}] {};
    \node[ellipse, draw, inner xsep=2pt, inner ysep=-10pt, fit=(2) (6), label={$i_2$}] {};
    \node[ellipse, draw, inner xsep=2pt, inner ysep=-10pt, fit=(3) (7), label={$i_3$}] {};
    \node[ellipse, draw, inner xsep=2pt, inner ysep=-10pt, fit=(4) (8), label={$i_4$}] {};

  \path[every node/.style={font=\sffamily\small}]
    (4) edge [-,dashed, line width=2pt] node[left] {} (5)
    (1) edge [-,line width=2pt] node[left] {} (6)
    (2) edge [-,dashed, line width=2pt] node[right] {} (7)
    (3) edge [-,line width=2pt] node[right] {} (8);
\end{tikzpicture}
\caption{Each line represents a pairwise exchange yielding utilities $u$ or $Tu$. Two solid-line exchanges result in allocation $Y$, while two dashed-line exchanges result in allocation $Z$.}
\end{figure}

If an $\{i,j\}$-allocation $(X_i, X_j)$ blocks $Y$ (so $v_j(X_j) > v_j(Y_j) \geq u_j$), then the utilities from blocking must be bounded above by $Tu$ ($v_{i}(X_{i}) < (Tu)_{i}$ and $v_{j}(X_{j}) < (Tu)_{j}$). For otherwise, $v_i(X_i) > (Tu)_i$ would contradict the definition of $Tu$, and the equality $v_i(X_i) = (Tu)_i$ holds only if $v_j(X_j) = u_j$ (as explained above using Lemma~\ref{lem:subsetofagents}), but it would contradict $v_j(X_j) > u_j$. Thus, if there is blocking of $Y$, $i_2$ and $i_4$ must be the blocking agents. The blocking of $Y$ leaves agents $i_1$ and $i_3$ to consume their endowments, and results in the allocation $(\omega_{i_1}, X_{i_2}, \omega_{i_3}, X_{i_4})$.

The blocking of $Y$ is followed by a sequence of forward-looking blocks, resulting in the allocation $Z$. For instance, the blocking by $\{i_2, i_4\}$ is succeeded by another blocking done by $\{i_2, i_3\}$ and defined by $((Tu)_{i_2}, u_{i_3})$. The resulting allocation is $(\omega_{i_1}, (Tu)_{i_2}, u_{i_3}, \omega_{i_4})$. Subsequent blocking by $((Tu)_{i_4}, u_{i_1})$ yields another allocation: $(u_{i_1}, (Tu)_{i_2}, u_{i_3}, (Tu)_{i_4})$, which we have called $Z$.

Therefore, the allocation set $\{Y,Z\}$ satisfies the following stability condition: any block of an allocation in $\{Y,Z\}$ is followed by a sequence of forward-looking blocks, resulting in another allocation in the same set $\{Y,Z\}$. A set of allocations with this property is termed \emph{locally absorbing}. While the entire set of allocations is locally absorbing, we identify a non-trivial locally absorbing set. A singleton set of a core allocation is locally absorbing, as is a von-Neumann-Morgenstern's external stable set where any allocation outside the set has a blocking that results in an allocation within the set.

A locally absorbing set may, in general, include more than two allocations. For any even-length cycle $C=\{i \to j \to \dots, i\}$, there are two $C$-allocations, where each $i \in C$ receives either $u_i$ or $(Tu)_i$. If agents in $A_2$ are divided into two even-length cycles, the locally absorbing set we construct would have four allocations. The locally absorbing sets that we find are minimal, in the sense that a proper subset of allocations is not a locally absorbing set.

Our arguments so far require all cycles to have even numbers of agents.  We later show that this condition of even length cycles holds for any 2-copy replica economy. We can extend our result to any replica economy with an even number of copies, and we could formulate an approximation result to deal with odd copies, but we do not work out the details of such a result here.

\subsubsection{A locally absorbing set}\label{sec:BSdefinition}

As discussed above, the idea behind a locally absorbing set is that any blocking of an allocation in this set is followed by a sequence of forward-looking blocks, resulting in another allocation in the set. We now formalize this idea.

A coalition $S\subset N$ \df{blocks} an allocation $X$ if there exists an
$S$-allocation $(Y_i)_{i \in S}$ such that $v_i(Y_i) > v_i (X_i)$ for all $i\in S$. The blocking by $(Y_i)_{i \in S}$ defines a new allocation $Y$ in
which: a member of the blocking coalition $S$ gets $Y_i$, the agents who are not in $S$ but were trading with an agent in $S$ under $X$ consume their endowments, and the remaining agents receive the same as in $X$. In other words, $Y_h=\w_h$ for those agents $h\notin S$ who were receiving $X_h$ through a $T$-allocation with $S \cap T\neq \emptyset$; and $Y_h=X_h$ for all agents $h\notin S$ for whom $X_h$ is obtained through a $T$-allocation with $S \cap T=\emptyset$.

An allocation $Y$ is \df{dominated} by an allocation $Z\neq Y$ if there is a sequence of allocations $Z^1, Z^2 \dots, Z^n$ such that:
\begin{enumerate}
\item $Z^1=Y$, $Z^n=Z$, and
\item ($\forall m=1, 2, \dots, n-1$) 
    \begin{itemize}
    \item A coalition $S_m$ blocks $Z^m$, which results in $Z^{m+1}$,
    \item $Z^{m+1}_{i_m} = Z^{m+2}_{i_m} = \cdots = Z^n_{i_m}$ for all $i_m \in S_m$.
    \end{itemize}
\end{enumerate}

A collection of allocations $\mathcal{S}$ is \emph{locally absorbing} if each $X \in \mathcal{S}$ is either unblocked or any blocking results in an allocation dominated by one in $\mathcal{S}$. We can similarly define \emph{a pairwise locally absorbing set} when only pairwise exchanges are allowed.

A locally absorbing set always exists. When the core is non-empty, any subset of the core is locally absorbing. If the core is empty, the set of all allocations, or the individually rational allocations, is locally absorbing. The reason is that blocking any allocation in the set results in another allocation that is dominated by yet another allocation in the set. We are interested in minimal locally absorbing sets, and the non-emptyness of the core yields the smallest possible such sets.

The question is whether a non-trivial (small) locally absorbing set exists, even when the core is empty. One possibility is to construct a locally absorbing set through an iterated elimination of allocations that do not dominate other allocations. Initially, we restrict the set to individually rational allocations and then eliminate any allocations that do not dominate any other allocation. Then we further eliminate allocations that only dominate the allocations eliminated in the previous round.\footnote{For the left-right shoe economy in Example~\ref{ex:shoes}, this construction obtains a minimal locally absorbing set $\{X,Y,Z\}$.} While repeating this procedure results in a smaller absorbing set, it can still be quite large. Alternatively, we can begin with a singleton set of a core allocation or a vNM stable set, both of which would be locally absorbing, but the core might be empty, and a stable set may not exist. Therefore, we aim to explore a different construction using a fixed point $u$ of $T^2$ with the condition $u \leq Tu$.

\subsubsection{The construction of a pairwise locally absorbing set}\label{sec:BSresults}

We restrict attention to pairwise exchanges, and therefore consider the function $T=T_k: U \to U$ with $k=2$. Then we take a fixed point $u$ of $T^2$ with $u \leq T u$, and  partition the set of agents such that $A_1= \{ i : u_i = (Tu)_i\}$ and $A_2= \{i : u_i < (Tu)_i\}$. The preallocation defined by $u$ (which is $\{ v_i^{-1}(u_i) : i \in A\}$) is an $A_1$-allocation (Lemma~\ref{lem:subsetofagents}). It remains to find an allocation for agents in $A_2$.

\begin{lemma}\label{lem:cycle_construction}
For any $i \in A_2 = \{i : u_i < (Tu)_i\}$, there exists a unique agent $j \in A_2 \backslash \{i\}$ such that $((Tu)_i, u_j)$ is an $\{i,j\}$-allocation, i.e., $v^{-1}_i((Tu)_i)$ and $v^{-1}_j(u_j)$ partition $\w_i \cup \w_j$.
\end{lemma}

We consider a directed graph such that each agent is a node, and an agent $i\in A_2$ has exactly one outgoing directed edge $i\to j$, where $j\in A_2$ is the agent identified in Lemma~\ref{lem:cycle_construction}. Thus when $i\in A_2$ and $i\to j$ then  $((Tu)_i, u_j)$ is an $\{i,j\}$-allocation. When $i\in A_1$, the graph features a self-loop, $i\to i$.

Let $\mathcal{P}$ be a partition of the set $A_2$, defined by cycles in the graph. That means, each partition set $C=\{i_1, i_2, \dots\} \in \mathcal{P}$ comprises agents in a cycle $i_1 \to i_2 \to \dots \to i_1$, which subsequently means that, for example, $((Tu)_{i_1}, u_{i_2})$ is an $\{i_1, i_2\}$-allocation.

If a cycle $C$ has an even number of agents, we may define two $C$-allocations depending on whether the odd-numbered, or the even-numbered, agents get the allocation that corresponds to $Tu$. Let $Y^C$ be a $C$-allocation with $v_{i_1}(Y^C_{i_1}) = (Tu)_{i_1}$, and similarly for all odd-numbered agents, and $v_{i_2}(Y^C_{i_2}) = u_{i_2}$ and similarly for all even-numbered agents. On the other hand, let $Z^C$ be a $C$-allocation with $v_{i_2}(Z^C_{i_2}) = (Tu)_{i_2}$ and similarly for all even-numbered agents, and $v_{i_1}(Z^C_{i_1}) = u_{i_1}$ and similarly for all odd-numbered agents. Note that the definitions of these allocations require that $C$ have an even number of agents.

If every cycle $C \in \mathcal{P}$ has an even number of agents, then we define a set of allocations, 
\begin{equation}
    F_u=\{ X : X_i = v_i^{-1}(u_i), \forall 
 i\in A_1, \text{ and } X^C \in \{Y^C, Z^C\}, \forall C\in\mathcal{P} \}.
    \label{eq:stable_set}
\end{equation}

\begin{proposition}\label{prop:farsighted_stability}
Let $u\in U$ be a fixed point of $T^2$ such that $u\leq Tu$. Consider the partition $\mathcal{P}$ of the set of agents $A_2 = \{ i : u_i < (Tu)_i\}$. If every cycle $C \in \mathcal{P}$ has an even number of agents, then $F_u$ is a pairwise locally absorbing set.
\end{proposition}

The process of constructing cycles to partition the set of agents coincides with the approach in \cite{tan1991necessary} to the roommate problem. From any fixed point $u$ of $T^2$ alongside its image $Tu$ (with $u \leq Tu$), we identify for each agent $i$ at most two trade partners. We identify precisely one exchange per partner, yielding utilities $u_i$ or $(Tu)_i$. Thus by limiting exchanges to those yielding utilities $u$ or $Tu$, we essentially transform the original exchange economy into a roommate problem.

\cite[Proposition 3.2]{tan1991necessary} shows that, if all cycles with at least two agents have even lengths, a roommate matching defined by the cycles is stable. The situation is somewhat different in our model. There can be a blocking coalition $\{i, j\}$ with a $\{i,j\}$-allocation $(X_i, X_j)$ such that $u_i < v_i(X_i) < (Tu)_i$ and $u_j < v_j(X_j) < (Tu)_j$. The exchange economy model in our paper is more general than a traditional roommate problem, as two agents can have many potential exchanges, resulting in utilities different from $u$ or $Tu$. The Pareto frontier of achievable utilities by a pair of agents may be non-trivial. This makes our problem quite different from Tan's, and the use of the $T$ algorithm is crucial in our arguments.

Fortunately, when agents $i,j$ block an allocation defined by cycles in $\mathcal{P}$, their utilities from blocking are strictly bounded above by $(Tu)_i$ and $(Tu)_j$. This property lets us identify a sequence of forward-looking blocks that lead to a different allocation defined by the cycles. Even if a part of Tan's result -- no blocking -- does not hold, there still exists a sense of stability for the set of allocations defined by the cycles.

\begin{proof} (Proof of Proposition~\ref{prop:farsighted_stability})
If all allocations in $F_u$ are unblocked, the proof is complete. Suppose agents $i$ and $j$ block an allocation $X \in F_u$ with a $\{i,j\}$-allocation $(X'_i, X'_j)$. Let $\overline{X}$ be the resulting allocation after the block.

We observe that $i \neq j$, because $v_i(X'_i) > v_i(X_i) \geq u_i \geq v_i(\omega_i)$. The allocation $X$ is individually rational, so it cannot be blocked by an agent alone. 

First, note that $i, j \in A_2$. Agent $i$ can obtain at most $(Tu)_i$ if other agents' utilities must be at least $u$. So, if $i \in A_1$, then $v_i(X'_i) > v_i(X_i) \geq u_i = (Tu)_i$ while $v_j(X'_j) > u_j$, which would contradict the definition of $Tu$. 

Second, the utilities from pairwise blocks are strictly bounded above by $Tu$: $u_i < v_i(X'_i) < (Tu)_i$ and $u_j < v_j(X'_j) < (Tu)_j$. We already noted that $v_i(X_i) < v_i(X'_i) \leq (Tu)_i$. Since every agent obtains a level of utility either $u$ or $Tu$ in the allocation $X \in F_u$, we have $v_i(X_i) = u_i$. A similar argument shows that $v_j(X_j) = u_j$. Moreover, as the utility function $v_i$ is injective, agent $i$ obtains $(Tu)_i$ only from a pairwise exchange with a particular trading partner, say $i'$, where $i'$ obtains $u_{i'}$. The exchange $\{X'_i, X'_j\}$ is not such an exchange because $v_j(X'_j) > v_j(X_j) = u_j$. Thus, $v_i(X'_i) < (Tu)_i$. A symmetric argument shows that $v_j(X'_j) < (Tu)_j$. 

Next, we identify a sequence of pairwise blocks to $\overline{X}$ that ultimately leads to an allocation in $F_u$.

Suppose that the blocking agents $i, j \in A_2$ belong to two different cycles $C_i, C_j \in \mathcal{P}$, respectively. We label agents in $C_i$ such that agent $i$ is numbered as $i_1$ in the cycle $C_i = \{ i_1 \to i_2 \to \dots \to i_n \to i_1\}$. Then, the odd-numbered agents in $C_i$ obtain the utility $u$ in the initial allocation $X$, and the even-numbered agents in $C_i$ obtain $Tu$. Moreover, in the allocation $\overline{X}$ resulting from the blocking of $X$ by $\{i,j\}$, we have
$$\overline{X}_{i_1} = X'_{i_1}, \quad \overline{X}_{i_n} = \omega_{i_n}, \quad \text{and} \quad \overline{X}_{i'} = X_{i'} \text{ for any other } i' \in C_i.$$
It remains true in $\overline{X}$ that odd-numbered agents in $C_i$ obtain utilities strictly less than $Tu$, and even-numbered agents in $C_i$ obtain utilities equal to $u$, with one exception of agent $i_n$ who obtains $v_{i_n}(\omega_{i_n}) < u_{i_n}$.\footnote{Since $i_n$ is in a cycle, $i_n \in A_2$, which means the agent obtains the utility $u_{i_n}$ in an exchange with a partner, say $i'$, who obtains $(Tu)_{i'}$.} We can then identify a sequence of blocks: $\{i_{n-1}, i_n\}$ blocks $\overline{X}$ so that they obtain $((Tu)_{i_{n-1}}, u_{i_n})$. If $n = 4, 6, \ldots$ (the even length cycle $C_i$ is longer than 2), then the agent $i_{n-2}$ is now left to consume her endowment, so $\{i_{n-3}, i_{n-2}\}$ blocks and obtain $((Tu)_{i_{n-3}}, u_{i_{n-2}})$. The sequence of blocks continues until $\{i_1, i_2\}$ blocks and increase their utilities from $(v_{i_1}(X'_{i_1}), v_{i_2}(\omega_{i_2}))$ to $((Tu)_{i_1}, u_{i_2})$. In the end, the sequence of blocks results in a $C_i$-allocation, either $Y^{C_i}$ or $Z^{C_i}$ as defined above in Proposition~\ref{prop:farsighted_stability}. We omit a similar sequence of objections by agents in the cycle $C_j$, ensuring the final allocation is in $F_u$.

Suppose, on the other hand, that the blocking agents $i,j \in A_2$ are in the same cycle $C \in \mathcal{P}$. We number the agents in $C=\{i_1 \to i_2 \to \dots i_n \to i_1\}$ such that $i = i_1$. Then, an agent $i' \in C$ is odd numbered if and only if $v_{i'}(X_{i'}) = u_{i'}$. We observed above that $v_j(X_j) = u_j$, so $j$ must be an odd-numbered agent, say $j=i_k$ with an odd number $k$. Then, we have two \emph{chains} (instead of cycles) $C_i=\{i \to i_2 \to \dots \to i_{k-1}\}$ and $C_j = \{ j \to i_{k+1} \to \dots \to i_n\}$. We take the chain $C_i$ and observe that, in the allocation $\overline{X}$,  
$$\overline{X}_{i} = X'_{i}, \quad \overline{X}_{i_{k-1}} = \omega_{i_{k-1}}, \quad \text{and} \quad \overline{X}_{i'} = X_{i'} \text{ for any other } i' \in C_i.$$
In the allocation $\overline{X}$, the odd-numbered agents in $C_i$ obtain utilities strictly less than $Tu$, and the even-number agents obtain utilities equal to $Tu$, with the exception of agent $i_{k-1}$ who obtains $v_{i_{k-1}}(\omega_{i_{k-1}}) < u_{i_{k-1}}$. The same argument as above when $C_i$ was a cycle continues to apply when $C_i$ is a chain. It finds a sequence of pairwise blocks by $\{i_{k-2}, i_{k-1}\}, \dots, \{i_1, i_2\}$ after which each odd numbered agent $i'$ obtains $(Tu)_{i'}$ and each even numbered agent $i''$ obtains $u_{i''}$. Similarly, we can find a sequence of pairwise blocks by agents in $C_j$, and the final allocation is in $F_u$.
\end{proof}

\subsubsection{A replica economy}\label{sec:replica}

The assumption in Proposition~\ref{prop:farsighted_stability} regarding all cycles having an even number of agents may not hold in general such as the shoe economy (Example~\ref{ex:shoes}). 
However, since each agent $i$ has a unique $j$ such that an $\{i, j\}$-allocation obtains $((Tu)_i, u_j)$, cycles in a 2-copy replica economy -- two copies of each agent ``type'' -- can have only even lengths.\footnote{\cite{chiappori2014roommate} also considers a 2-copy replica economy to establish a stable outcome for the roommate problem with transferable utilities.}

Formally, given a baseline economy $E = (O, \{ (v_i, \omega_i) : i \in A\})$, a 2-copy replica economy $\overline{E} = (\overline{O}, \{(v_i, \omega_i) : i \in \overline{A}\})$ is defined by
\begin{itemize}
\item $\overline{A} = A' \cup A''$, where $A'=A''=A$. We call $i' \in A'$ and $i'' \in A''$ the first and the second copies of $i \in A$.

\item For $o \in \omega_{i}$, we have $o' \in \omega_{i'}$ and $o'' \in \omega_{i''}$. Moreover, $\overline{O} = O' \cup O''$, where $O' = \cup_{i' \in A'} \omega_{i'}$ and $O'' = \cup_{i'' \in A''} \omega_{i''}$. 

\item $v_{i}=v_{i'}=v_{i''}$. Two copies of the same object are utility-indistinguishable, and consuming a duplicate copy adds no utility.
\end{itemize}

Given a baseline economy $E$, let $u$ be a fixed point of $T^2$ with $u \leq Tu$. Since each agent $i$ in the baseline economy has a unique trading partner $j$ such that an $\{i, j\}$-allocation obtains $((Tu)_i, u_j)$, the replica economy has a (feasible) allocation such that, for every agent (type) $i$, one copy obtains $(Tu)_i$ and the other copy obtains $u_i$. 


A straightforward extension of a fixed point $u$ from the baseline economy is a fixed point $\overline{u}$ of the replica economy: $\overline{u}_{i'}=\overline{u}_{i''} = u_i$ for copies $i', i''$ of $i$. The key observation is that
\[\begin{split}
(T\overline{u})_{i'} =  \max B_{i'}^{k=2}(\overline{u}) 
 = & \max \{ \hat{u}_{i'} : \text{$\exists j'' \in A''$ with an $\{i', j''\}$-allocation $(X_{i'}, X_{j''})$}\\
&\text{such that $v_{i'}(X_{i'})=\hat{u}_{i'}$ and $v_{j''}(X_{j''})\geq \overline{u}_{j''}$}\}\\
 = & (Tu)_i.
\end{split}\]

The equalities hold even though agent $i'$ is restricted to trade only with agents in $A''$. Since each copy of the same agent type has the same endowment, it is without loss to assume that the first copy of an agent type trades only with the second copy of any agent type. 

Moreover, each cycle $C \in \mathcal{P}$ in the baseline economy has a straightforward extension to an even-length cycle $\overline{C}$ in the replica economy. For example, a cycle $i \to j \to k \to i$ is extended to $i' \to j'' \to k' \to i'' \to j' \to k'' \to i'$. The first copies and the second copies alternate because exchanges are restricted, without loss, to be between the first copy of a type and the second copy of another type.

\section{Conclusion}

Our paper explores stable allocations in discrete exchange economies, employing two approaches inspired by \cite{shapley1974cores}'s work on obtaining the weak core in the housing market. Firstly, we utilize Scarf's balancedness condition for a NTU game, investigating sufficient conditions on preferences for the exchange economy to define a balanced NTU and, therefore, have a nonempty weak core. Notably, this necessitates restrictive assumptions, specifically a categorical economy with additively separable and dichotomous preferences. Secondly, we extend Gale's Top Trading Cycle algorithm as an iteration of a Tarski function ($T^2$), applicable to any exchange economy under mild preference assumptions (strict and monotone preferences). We observe that a fixed point of $T$ (rather than $T^2$) is related to a stable allocation, but a fixed point may not exist. Accordingly, we work with a fixed point of $T^2$ (which always exists) and explore how it relates to a different notion of stability.

\appendix

\section{Proof of Theorem~\ref{thm:categorical_economy}}
\label{sec:proofs}


The weak random core is nonempty as proven in Proposition~\ref{prop:fractional_core}. We take a random core allocation $\sigma$ and define target utilities $(t_i)_{i \in A}$ as the floor of $v_i(\sigma)$ for $i \in A$. As explained after Theorem~\ref{thm:categorical_economy}, our final step is to round the fractional matrix $P$ defined by the random allocation $\sigma$ to obtain a deterministic allocation $X$ such that $v_i(X_i) \geq t_i$ for $i \in A$.



We consider a rounding algorithm that operates over any matrix $B$ with $\abs{A}$ rows and $\sum_k \abs{O^k}$ columns, augmented by a column vector $\tilde{t}$. Initially, $(B | \tilde{t})= (P | t)$ and satisfies:\begin{enumerate}
    \item All entries of $B$ are in $[0,1]$,
    \item\label{it:matrixtwo} Each column sums up to 1.
    \item\label{it:matrixthree}
    $\sum_{k} \sum_{j \in O^k} B(i,j)  \geq \tilde{t}_i$, and 
    \item\label{it:matrixfour} If $b_{i,j}>0$, then $v_i^k(j)=1$ for the category $k$ of object that $j$ is.
\end{enumerate}

The algorithm has two major subroutines: preprocessing and rounding.

\begin{enumerate}
    \item{\textbf{Preprocessing}}
    
    We apply the following procedures to $(B|\tilde{t})$ in no particular order until there is no further update. The number of rows may increase, so we refer to $i$ as a row rather than an agent.

\begin{enumerate}

    \item \label{pre_processing_one} Remove any row $i$ with $\tilde{t}_i=0$ (row $i$ gets no object), and any column $j$ with $\sum_{i} B(i,j) = 0$ (no row gets object $j$).
    
    \item \label{pre_processing_two} For column $j$, if there exists a unique $i$ such that $B(i,j) > 0$, then increase $B(i,j)$ to $1$, and reduce $B(i,j')=0$ for any $j'$ in the same category as $j$.
    \item \label{pre_processing_three} Replace any entry $B(i,j)=1$ with $0$ and reduce $\tilde{t}_i$ by 1 ($j$ is assigned to the agent associated with row $i$).      
    	\item \label{pre_processing_four} For row $i$ and category $k$ such that $\sum_{j\in O^k} B(i,j) =1$ (a guaranteed utility of 1 from category $k$) and $\tilde{t}_i > 1$ (additional guaranteed utility), divide row $i$ into two rows. One row inherits category $k$ ($\{B(i,j)\}_{j \in O^k}$), the other row inherits other categories ($\{B(i,j)\}_{j \notin O^k}$), and the remaining entries of the two rows are zero. Accordingly, expand the target-utility vector $\tilde{t}$ such that $\tilde{t}_i$ is divided into two elements $1$ and $\tilde{t}_i -1$.
	
	\end{enumerate}

    \item{\textbf{Rounding}}
    
    $(B|\tilde{t})$ satisfies:
        \begin{enumerate}
        \item \label{it:matrixone_revised}All entries of $B$ are in $[0,1)$,
        \item\label{it:matrixtwo_revised} 
        For each $j$, 
        $0< \sum_i B(i,j) \leq 1$ and there exist at least two rows $i, i'$ with $B(i, j) > 0$ and $B(i', j) > 0$.
        \item\label{it:matrixthree_revised} 
        $\sum_k \sum_{j\in O^k} B(i,j) \geq \tilde{t}_i \in\{1, 2, \dots, K-1\}$.
        \footnote{$\tilde{t}_i \neq 0$ because we applied preprocessing \ref{pre_processing_one}; $\tilde{t}_i \neq k$ because we applied preprocessing \ref{pre_processing_four}.}
        \end{enumerate}

    Here we consider a graph $(V, E)$ such that $V$ is the set of non-zero elements of $B$ and $E$ is the set of pairs of vertices that are in the same row or column. 

    We claim that there exists a cycle. Take any $(i,j)$ with $B(i,j) > 0$. By property~ \ref{it:matrixtwo_revised}, we find $i' \neq i$ such that $B(i', j) > 0$. Since the guaranteed utility $\tilde{t}_{i'}\geq 1$ (property~ \ref{it:matrixthree_revised}), there exists $j' \neq j$ such that $B(i', j') > 0$. Then, by  property~\ref{it:matrixtwo_revised}, we subsequently find $i'' \neq i'$ such that $B(i'', j') > 0$. When some agent appears twice in this construction of a path, a cycle $(i_1, j_1), \dots, (i_M, j_M)=(i_1, j_1)$ is found. Label the entries in the cycle such that 
    each odd-numbered entry is in the same column as the next entry and in the same row as the preceding one.

    We add $\epsilon > 0$ for each odd-numbered entry, and subtract $\epsilon$ from the next entry. The row and column sums remain the same. Choose the largest $\epsilon$ such that either (i) one entry in the cycle becomes integral, i.e., $0$ or $1$, or (ii) there exists a row $i$ and category $k$ such that the constraint $\sum_{j\in O^k} B(i,j) v_i^k(j) \leq 1$ has become newly binding.
\end{enumerate}

We iterate the preprocessing and rounding subroutines. Preprocessing weakly decreases the number of fractional elements of $B$, even if the size of the matrix may increase. Rounding either (i) decreases the number of fractional elements by at least one, or (ii) creates a new pair, consisting of a row $i$ and a category $k$, such that $\sum_{j\in O^k} B(i,j) = 1$; and such pair forms a new row of $(B| \tilde{t})$ in the follow-up preprocessing step~\ref{pre_processing_four}. Hence, the algorithm terminates in at most $\abs{P} + (|A| \times (K-1))$ rounds. When it terminates, the matrix $B$ is integral.

Finally, we identify an allocation from the integral matrix $B$ and the assignments by preprocessing~\ref{pre_processing_three} given to agents associated with the rows. The target utility is achieved because of property~\ref{it:matrixthree_revised} of the matrix $B$ and the sum of elements of $\tilde{t}$ associated with each agent has decreased by one only upon an assignment of an acceptable good by preprocessing~\ref{pre_processing_three}. Each agent consumes at most one acceptable good of each category $k$ because, when the algorithm runs, at most one row $i$ associated with the agent has $\sum_{j \in O^k} B(i,j) > 0$, until a possible assignment of an object $j$ in category $k$ by preprocessing~\ref{pre_processing_three}, after which every row $i$ associated with the agent has $\sum_{j \in O^k} B(i,j) = 0$.

We have established an allocation in which all agents attain their target utilities. This completes the proof, as discussed following Theorem~\ref{thm:categorical_economy}. 

\section{Gains from trade}\label{sec:gainsfromtrade}

Apart from categorical economies with dichotomous preferences, other conditions can also ensure a balanced NTU game.

An economy satisfies \df{gains from trade} if trades within a larger coalition can achieve at least the utilities that smaller constituent coalitions can achieve. Formally, for any 
two coalitions $S$ and $S'$, any  $S$-allocation $\{X_i : i\in S\}$ and any $S'$-allocation $\{X'_i : i\in S'\}$, there exists an $S\cup S'$-allocation $\{Y_i : i \in S\cup S'\}$ with $v_h(Y_h)\geq \min\{v_h(X_h),v_h(X'_h) \}$ for all $h\in S\cup S'$; where we define $X'_h=X_h $ when $h\in S'\setminus S$ or $h\in S' \setminus S$.

The assumption of gains from trade is violated in our Example~\ref{ex:shoes}. Consider the coalitions $S=\{1,2\}$ and $S'=\{2,3 \}$, with the $S$-allocation $(X_1,X_2)=((\ell_1,r_2),(\ell_2,r_1))$, and the $S'$-allocation $(X'_2,X'_3)=((\ell_2,r_3),(\ell_3,r_2))$. It is clear that any allocation in $S\cup S'$ that gives $1$ at least the utility from $X_1$ must give $1$ the bundle $X_1$. This rules out that 3 gets $r_2$, so 3 would have to receive $(\ell_1,r_3)$ in order to be as well off as in $X'_3$. This is, however incompatible with $1$ getting $\ell_1$.

While the ``gains from trade'' condition is restrictive, it can hold when agents' endowments are specialized for each kind of product (production specialization or having exclusive access to suppliers of each kind). A larger coalition ensures that each member obtains a more diversified consumption, which can be desired.

\begin{theorem}\label{thm:gainstrade}
If an economy has gains from trade, and all agents' utility functions are injective, then the weak core is nonempty.
\end{theorem}

We present two proofs. The first proof proceeds by showing that Scarf's sufficient condition is satisfied. The second proof shows that an economy as given in the theorem is associated to an ordinally convex NTU game. These are known to have nonempty weak core (\cite[Theorem 12.3.3]{peleg2007introduction}).

\begin{proof}
We present two proofs. The first proof shows that, if an economy has gains from trade, the induced NTU game $(A,V)$ is balanced and therefore, has a nonempty weak core. 

Suppose that $\S$ is a balanced collection of coalitions and $u \in \cap_{S \in \S} V(S)$. Take any $S,S' \in \S$. We know there exists $S$-allocation $X$ and $S'$-allocation $X'$ such that $u_i \leq v_i(X_i)$ for $i\in S$ and $u_i \leq v_i(X'_i)$ for $i \in S'$. By gains from trade, there exists $S \cup S'$-allocation $Y$ such that $v_h(Y_h) \geq \min\{v_h(X_h),v_h(X'_h)\} \geq u_i$  for $h \in S \cup S'$. We can now apply gains from trade to coalitions $T=S\cup S'$ and $T' \in \S$ to get a $T \cup T'$- allocation in which every agent $i \in T \cup T'$ gets at least $u_i$. Eventually, we'll get an allocation in which each agent $i \in A$ gets at least $u_i$ and so we get that $u \in V(A)$, and so the game $(A, V)$ is balanced.

Next, we present the second proof, based on the theory of ordinally convex games: The defined game is ordinally convex if for all $S, S' \subseteq A$, 
\begin{equation}
V(S) \cap V(S') \subseteq V(S\cap S') \cup V(S \cup S'). 
\label{convex_game}
\end{equation}
An ordinaly convex game has nonempty weak core (\cite[Theorem 12.3.3]{peleg2007introduction}). 

We proceed to show that the game $(A, V)$ is ordinally convex. Consider $S, S' \subseteq A$. If $S \subseteq S'$, then \eqref{convex_game} holds trivially because $V(S) = V (S \cap S')$ and $V(S') = V(S\cup S')$. A similar conclusion holds if $S' \subseteq S$. Thus, we suppose that $S\cap S'$ is a strict super set of $S$ and $S'$. 

Take $u \in V(S) \cap V(S')$. There are $S$-allocation $\{Y_i : i \in S\}$ and  $S'$- allocation $\{Y'_i : i\in S'\}$ such that $u_i \leq v_i(Y_i)$ for $i\in S$ and $u_j \leq v_j(Y'_j)$ for $j \in S'$. If $S$ and $S'$ are disjoint, then $\{Y_i : i\in S\} \cup \{Y'_j : j \in S'\}$ is an $S \cup S'$-allocation. Then, $u \in V(S \cup S')$. Consider the other case of $S \cap S' \neq \emptyset$. If $Y_i = Y'_i$ for $i \in S \cap S'$, then by no indifferences, no object in $\{Y_i : i \in S \cap S'\}$ ($= \{ Y'_i : i \in S \cap S'\}$) is from the endowments of $S \backslash S'$ or $S'\backslash S$. Thus, $\{Y_i : i \in S\cap S'\}$ is an $S\cap S'$-allocation. It follows that $u \in V(S \cap S')$. On the other hand, if for some $j \in S \cap S'$, $v_j(Y_j) < v_j(Y'_j)$ (we omit the other case of $v_j(Y_j) > v_j(Y'_j)$), then gains from trade implies that there exists an $S \cup S'$- allocation $\{ Z_i: i \in S \cup S'\}$ with $v_i(Z_i) \geq \min\{v_i(Y_i), v_i(Y'_i)\}$ for all $i \in S \cup S'$. Hence, $u_i \leq  v_i(Z_i)$ for all $i \in S \cup S'$, which implies $u \in V(S\cup S')$.
\end{proof}

The second proof is not applicable to a categorical economy with additively separable and dichotomous preferences because it does not define an ordinally convex game. For example, consider an economy with three agents $\{1, 2, 3\}$ such that agents $1$ and $2$ consider each other's endowment  acceptable, and agents $2$ and $3$ consider each other's endowment acceptable. For any coalition $\{i, j\}$ with $(i=1, j=2)$ or $(i=2, j=3)$, $(1, 1, 1) \in V(\{i,j\})$ because there exists a $\{i,j\}$-allocation $\{X_i=\w_j, X_j = \w_i\}$ such that $1 \leq u_i(X_i)$ and $1 \leq u_j(X_j)$. However, $(1, 1, 1) \neq V(\{1, 2, 3\})$ because agents 1 and 3 consider only agent 2's endowment acceptable. The convexity condition \eqref{convex_game} does not hold.

\section{Proofs regarding the $T$ algorithm}\label{sec:pfsalgo}

\subsection{Proof of Lemma~\ref{lem:TTC_extension}}

We prove the lemma by induction.
\begin{itemize}
    \item ($r=1$) The utility $T\underline{u}$ identifies an $A_1$-allocation. Formally, $\{v^{-1}_i((T\underline{u})_i) : i \in A_1\}$ is an $A_1$-allocation such that each $i \in A_1$ gets her most preferred house. Consequently, for $i\in A_1$, $(T\underline{u})_i = (T^{2}\underline{u})_i$, and   \eqref{eqn:semi-monotonic-T} implies that $(T\underline{u})_i= (T^{2}\underline{u})_i= (T^{3}\underline{u})_i=\dots$.
    
    \item ($r=2$) Given $T\underline{u}$, no agent in $A_1$ is willing to trade with an agent not in $A_1$. Hence, for any agent $i \notin A_1$, $(T^2\underline{u})_i$ is a house that is not allocated in the first round of TTC. Thus,  $T^3\underline{u}$ identifies an $A_2$-allocation such that an agent in $A_2$ gets her most preferred remaining house after the first round of TTC. Consequently, for $i \in A_2$, $(T^3\underline{u})_{i} = (T^{4}\underline{u})_{i}$, and \eqref{eqn:semi-monotonic-T} implies that $(T^3\underline{u})_{i} = (T^{4}\underline{u})_{i} = (T^5\underline{u})_{i} = \dots$.
    \item ($r > 2$) A proof is similar to the previous step, so we omit.
\end{itemize}

\subsection{Proof of Lemma~\ref{lem:Talgobasic}}
\paragraph{Part 1:} 

If $u\leq u'$ then $B^k_i(u')\subseteq B^k_i(u)$, which implies that $Tu'\leq Tu$. In turn this means that $T^2$ is monotone increasing. Let $\ul u = (v_i(\w_i))_{i\in A}$ and note that $\ul u\leq Tu$ for all $u$. In particular, $\ul u \leq T^2 \ul u$, which implies that the sequence $T^{2m}\ul u$ is monotone increasing. Since $U$ is finite, there is $m$ so that $u:=T^{2m}\ul u = T^{2(m+1)}\ul u$. Such $u$ is a fixed point of $T^2$. The preallocation defined by $u$ is individually rational because $u = T^2 u \geq \underline{u}$. 

Note $\ul u \leq T \ul u$, 
and $T^2$ is monotone increasing. Thus, $T^{2m}\ul u\leq T^{2m+1} \ul u$, and we have $u\leq T u$ for the above fixed point $u$.

\paragraph{Part 2:} 
When $u = Tu$, the set $A_1 = \{ i : u_i = (Tu)_i \}$ is equal to $A$. Then, Lemma~\ref{lem:subsetofagents}, which we prove next, implies that the preallocation defined by $u$ is an allocation.

The allocation defined by $u$ is individually rational because any point in the image of $T$ is individually rational. Similarly, if there is a coalition $S$ of size at most $k$ and an $S$-allocation $X$ such that $v_i(X_i) > u_i$ for $i \in S$, 
then this would violate $u_{i} = (Tu)_{i}$ for $i \in S$.

\subsection{Proof of Lemma~\ref{lem:subsetofagents}}

Take any $i \in A_1$. Since $(T^2u)_i = u_i = (Tu)_i$, we can find a coalition $S$ such that $i \in S$, $\abs{S} \leq k$, and has a $S$-allocation $\{Y_j : j\in S\}$ such that $v_i(Y_i)=u_i$ and $v_j(Y_j) \geq (Tu)_j$ for $j \in S \backslash \{i\}$. Observe that this $S$-allocation $\{Y_j : j\in S\}$ ensures $(Tu)_j$ for all $j \in S$, so $u_j = (T(Tu))_j \geq v_j(Y_j) \geq (Tu)_j$. Together with $u \leq Tu$, we obtain $u_j  = (Tu)_j = v_j(Y_j)$ for $j \in S$. Therefore, $S\subseteq A_1$. 

Now, consider another agent $i' \in A_1 \backslash \{i\}$. A similar argument allows us to find a coalition $S'$ with $i' \in S'$, $|S| \leq k$, and a $S'$-allocation $\{Y'_j : j \in S'\}$ such that $u_j = (Tu)_j = v_j(Y'_j)$ for all $j \in S'$. Thus, $S'$ is also a subset of $A_1$.

Because the utility functions are injective, we have $Y_j = Y'_j = v_j^{-1}(u_j)$ for $j \in S \cap S'$. Then, since $Y_j = Y'_j$ is available in both coalitions $S$ and $S'$, we have $Y_j = Y'_j \subseteq \cup_{h \in S \cap S'} \w_h$. Thus, $\{Y_j : j \in S \cap S'\}$ is a $S\cap S'$-allocation. 

Finally, due to the strict monotonicity, which follows from the monotonicity and injectivity, we can conclude that every object within the coalition $S \cap S'$ is used in the maximization $v_j(Y_j) = (Tu)_j = \max B^k_j(u)$. Therefore, the $S$-allocation $\{Y_j : j \in S \}$ can be partitioned into a $S\cap S'$-allocation and a $S \backslash S'$-allocation. Similarly, we also can partition the $S'$-allocation. Ultimately, we obtain a partition $\mathcal{P}$ of $A_1$ such that for each $S \in \mathcal{P}$, $\{v^{-1}_i(u_i) : i\in S\}$ is a $S$-allocation. Hence, $\{ v_i^{-1}(u_i) : i \in A_1\}$ is an allocation for the coalition $A_1$.

\subsection{Proof of Lemma~\ref{lem:cycle_construction}}

Since $i \in A_2$, $u_i = (T^2 u)_i < (Tu)_i$. Since $k=2$, there exists $j$ and an $\{i,j\}$-allocation $(X_i, X_j)$ such that $v_i(X_i) = (Tu)_i > u_i$ and $v_j(X_j) \geq u_j$. Clearly, $j\neq i$ because $(Tu)_i > u_i \geq v_i(\w_i)$ (some object in $X_i$ must be from $\omega_j$), where the last inequality holds because $u$ is in the image of $T$. 

In fact, $j \in A_2$. Lemma~\ref{lem:subsetofagents} showed that, in the preallocation defined by a fixed point $u$ of $T^2$, agents in $A_1$ trade among themselves. 
Since $v_j$ is injective, there is only one bundle that ensures $j$ achieves utility $u_j$, and consequently, $X_j \cap \w_i = \emptyset$. Then, by monotonicity, $X_j = w_j$, which contradicts $v_i(X_i) > u_i \geq v_i(\w_i)$.

The uniqueness of the agent $j$ holds because $v_i$ is injective. If agent $i$ achieves utility $(Tu)_i$ by consuming some objects owned by $j$, then no other pairwise trade offer $i$ the same utility $(Tu)_i$. Similarly, for agent $j$, the uniqueness of $i$ holds, because $v_j$ is injective. Agent $j$ achieves utility $u_j$ only if she consumes an object owned by agent $i$.

\bibliographystyle{ecta}
\bibliography{ec}

\end{document}